\documentclass[12pt]{amsart}
\usepackage{xcolor}
\usepackage{amssymb}
\usepackage{amsmath}
\usepackage{graphicx}
\newtheorem{thm}{Theorem}[section]
\newtheorem{cor}[thm]{Corollary}
\newtheorem{prop}[thm]{Proposition} 
\newtheorem{Lemma}[thm]{Lemma}
\newtheorem{remark}[thm]{Remark}

\usepackage{a4wide}

\usepackage[cp1250]{inputenc}

\newcommand{\R}{\mathbb{R}}

\begin{document}

\author{Aleksander \'{C}wiszewski}
\address{Faculty of Mathematics and Computer Science \newline Nicolaus Copernicus University \newline 87-100 Toru\'n, Poland}
\email{aleks@mat.umk.pl}

\author{S{\l}awomir Plaskacz}
\address{Faculty of Mathematics and Computer Science \newline Nicolaus Copernicus University \newline 87-100 Toru\'n, Poland}
\email{Slawomir.Plaskacz@umk.pl}

\title[Effects of Marine Reserve Creation]{Effects of  Marine Reserve Creation in Single Species and Prey-Predator Models}



\maketitle

\begin{abstract}
Single species fisheries and prey-predator models with marine protected areas (MPA's) are studied. The single species case is considered when the fishing effort is around the species extinction threshold and the influence of implementing MPA on catch quantity are studied. In the prey-predator fishery model, the situation with the fishing effort close to the predator extinction value is considered and the effects of implementing MPA are discussed with the focus on MPA sizes assuring an acceptable catch level (i.e. food security) and the sustainability of both the predator and prey populations.    
\end{abstract}
{\bf Keywords:} prey-predator model, single species model, overexploited fishery, Marine Protected Areas, conservation effects, yield effects, food security.

\section{Introduction}

In many part of the world's oceans it is not possible to control the fishing effort. Many fish species have been over-exploited. The alternative way to protect marine stocks relies on establishing an area where fishing is prohibited, i.e. the so-called Marine Protected Area (MPA). The effects of establishing an MPA has been considered by many authors mainly in the framework of one species (or metapopulation) models. In single-cohort  models the effect of reserves on the fish population has been found positive that is called a conservation benefit (comp. \cite{Matsuda}, \cite{Sanchirico}, \cite{Gerber}). There are some uncertainties concerning the yield improvements in the effect of establishing an MPA. 
It is generally regarded that if the fishery is initially over-exploited then yield increases after establishing an MPA (comp. \cite{Gerber}). A bioeconomic model considered in \cite{Costello SR} gives quantitative guidance of the optimal MPA size that maximizes catch. In \cite{Costello SR}, \cite{Gerber} and many other papers it was postulated to evaluate MPA effects in a multispecies model. In \cite{Takashina}, \cite{Misra}, \cite{Lv}  MPA effects were considered for prey-predator models. In \cite{Misra}, \cite{Lv} it was assumed that the predator does not occur in a reserve area. The full model with both species in the reserve area and in the non reserve area is considered in \cite{Takashina}. Takashina-Mougi-Iwasa show in \cite{Takashina} that establishing MPA can cause a reduction of prey population and even extinction of the prey. That conclusion is totally opposite to conservation benefits observed in single-species models. In the model considered in \cite{Takashina}, the predator is assumed to be a generalist that has an alternative food source other then the prey species in the model. The  population of alternative food source is not counted in the model as a quantity dependent to the predator environment volume.  

In the paper we describe  effects of establishing an MPA for a single species model and for a prey-predator model and we make a comparison of obtained results in both models. In the prey-predator model we assume that the predator is an individualist, i.e. it has no alternative food source. The other possible interpretation is that the prey population is a multispecies.  We shall study both conservation and yield effects. We shall study the models for various value of the parameter called {\em fishing effort} and describe its and other model parameters' role in the dynamics of these systems.\\
\indent The yield effect of establishing the MPA is different for over-exploited and under-exploited fishery. In models, the over-exploitation occurs when 
the fishing effort $E$ oversteps a certain threshold value. As the threshold value we take the value of the fishing effort that threatens the extinction of the fish species, which is red. In single species models only one extinction threshold value $E_{ext}$ occurs, i.e. such that
if the fishing effort $E$ is greater or equal to $E_{ext}$, then the fish population extincts. In the prey-predator model we have two threshold values for fishing effort. The predator extinct when the fishing effort $E$ exceeds 
the value $E_{y-ext}$ and the prey starts to disappear when $E$ oversteps the second threshold $E_{x-ext}$, which is more than $E_{y-ext}$ in the models we consider here. If the predator population is null in the prey-predator model then the model simply reduces to a single species (prey only) one. Therefore, the fish extinction threshold $E_{ext}$ in the single species model coincides with $E_{x-ext}$ of the prey extinction in the prey-predator model. It seems to be reasonable to assume that we have an over-exploited fishery in a prey-predator model if any species extincts. In our model it means that the fishing effort $E$ exceeds $E_{y-ext}$. Hence, the level of fishing effort corresponding to over-exploitation in the prey-predator model is smaller that the fishing effort corresponding to over-exploitation in the corresponding single species model. The main goal of the paper is to compare the yield and conservation  effects of establishing MPA in the prey-predator model and in the single species model for initially over-exploited fisheries. We shall address the problem arising when one wants to sustain the ecosystem. Namely we ask whether we can prevent the extinction of predator when $E$ exceeds $E_{y-ext}$ by imposing an MPA and what the price in terms of total catch, i.e. food security, would be.\\
\indent In the case of single species fisheries, it is natural to consider a logistic model of the form
\begin{equation}\label{log1}
\dot x=a x\left(1- x/K \right)- E q x,
\end{equation}
where $x(t)$ is the fish population at time $t$, $a>0$ is the population growth rate, $K>0$ is the environmental capacity for the fish, $E\geq 0$ is the fishing effort and $q>0$ is the catchability coefficient of the given fish species. The population dynamics in the presence of an MPA  is given by the following system
\begin{equation}\label{r1}
\left\{\begin{array}{l}
\displaystyle{\dot x_1= a x_1\left(1- x_1/(1-R)K\right) - E(R) q  x_1 + m \left((1-R)x_2-R x_1 \right)},\\
\dot x_2= a x_2 \left(1-  x_2/RK \right) - m \left((1-R)x_2-R x_1 \right),
\end{array}\right.
\end{equation}
where $x_i(t)$ are the fish population at time $t$ in area $i$ ($i=1,\,2$). Area $1$ and $2$ represent the fishing ground and MPA, respectively. We represent the fraction of of these areas as $1-R$ and $R$, respectively, where $R\in (0,1)$. We assume that fish migrate between the two areas at migration rate $m$. A similar model of the migration has been considered in \cite{Costello SR}, \cite{Takashina}. Establishing a marine reserve may cause the redistribution of fishing effort into the nonreserve area. In the case of non-redistribution, the fishing effort does not change in relation to the fishery size, i.e. $E(R)=E$. In the case of full redistribution we assume that the fishing effort is fully relocated from the MPA to the accessible fishery, i.e. $E(R)=E/(1-R)$, where $E$ is a fishing effort before establishing the MPA (comp.\cite{Apostolaki}). In the paper we study the case where the fishing effort is not redistributed.\\
\indent  If two species are present in a fishery, we shall study the following prey-predator fishery model 
\begin{equation}\label{pp1K}
\left\{\begin{array}{l}
\dot x = a x\left(1-x/K\right) - bxy/K - E q_x x,\\
\dot y = c xy/K - dy-Eq_y y,
\end{array}\right.
\end{equation}
where $x(t)$ is the prey population and $y(t)$ is the predator population at time $t$, $a$ and $K$ are as in \eqref{log1}, $b>0$ is the mortality death rate of prey caused by predator, $c>0$ is the coefficient describing the dependence of the predator growth rate on the density of prey and $d>0$ is the predator natural death rate. The catchability coefficients $q_x, q_y>0$ for both the species are interpreted as in the one species models. Here we assume additionally that $a>b>c>d$. Using the same modelling approach as in the single species model, after implementing an MPA of size $R\in (0,1)$ we get the following system
\begin{equation}\label{m1}
\left\{
\begin{array}{l}
\! \dot x_1 = a  x_1(1\!-\! x_1/(1 \!-\! R)K)- b x_1 y_1/(1\!-\!R)K-E q _x  x_1+ m_x ((1\!-\!R)x_2\!-\!R x_1),  \\
\! \dot x_2= a  x_2(1\!- \!x_2/RK)- b x_2 y_2/ R-m_x ((1\!-\!R)x_2\!-\! Rx_1),  \\
\! \dot y_1= c x_1 y_1/ (1\!-\!R)K- d  y_1-E(R)q_y y_1+m_y  ((1\!-\!R)y_2\!-\!Ry_1),\\
\! \dot y_2= cx_2 y_2/RK - d y_2 - m_y  ((1\!-\!R)y_2\!-\!Ry_1),
\end{array} \right.
\end{equation}
where $x_1(t), x_2(t)$ are the prey population (biomass) outside and inside the MPA, respectively, and analogously $y_1(t)$ and $y_2(t)$ stand for the predator population outside and inside MPA, respectively. Here $m_x, m_y$ are the mobility coefficients.\\
\indent Both in the single species and prey-predator models we ask if one can prevent the extinction of any species by implementing an MPA and how that MPA will influence the total catch. In the prey-predator model we will also study the total catch coming from the predator as well as how the predator extinction threshold $E_{y-ext}(R)$, understood as a minimal fishing effort that generate predators extinction, is raised by creating an MPA of size $R$ (securing against the predator's extinction).\\
\indent The prey-predator model with MPA \eqref{m1} is described by a system of four nonlinear ODE's, which makes studying its equilibrium points by use of explicit formulas impossible. This causes the main mathematical difficulty.
Nevertheless, we prove the uniqueness of positive equilibria (with all the four coordinates positive) for a quite general class of models -- see Theorem \ref{necessary-condition}. We also indicate a general algorithm leading to these equilibria -- see Remark \ref{exitstence-rem}. Moreover, in Theorem \ref{EQ1.1} we show the existence of the fishing effort threshold $E_{y-ext}(R)$ (depending on the MPA size $R$) that is bigger than $E_{y-ext}$ and below which MPA assures the existence of positive steady-states and above which there are no positive steady-states.
We will also see that after the implementation of MPA, it is again the predator that extincts as the first one as $E$ grows -- see Remark \ref{extinction-order}. 
Our results and simulations show that implementing an MPA, even when $E$ is close to $E_{y-ext}$, gives a significant 'safety margin' that secures the predator from extinction -- see Subsection 5.3.

The papers is organized as follows. In Section 2 we describe the single species fishery model and study the effects of implementation of MPA. The prey-predator fishery model without MPA is examined in Section 3. In Section 4 we deal with the prey-predator model with MPA. We show the uniqueness of positive steady-states and its existence in dependence on fishing effort as well as discus its stability. Section 5 is devoted to the interpretation and discussion of numerical simulations. In Section 6 we end the paper with concluding remarks referring to the initially stated problems.


\section{Single species model analysis}


In the section we consider the single species fishery model given by \eqref{log1}. We assume that the capacity is normalized, $K=1$, i.e.  the fish population  $x(t)$ is assumed to satisfy $0\leq x(t)\leq 1$ (one may think of $x(t)$ as the density of the species relative to the environment capacity). Clearly, the equilibrium point of (\ref{log1}) is given by $\bar{x}=\frac{a-Eq}{a}$. This equilibrium is positive if and only if
\begin{equation}\label{Eext}
E<E_{ext}:=\frac{a}{q}.
\end{equation}
It is clear that the positive equilibrium point $\bar{x}$ is globally asymptotically stable in the interval $(0,1)$. 
If the fishing effort $E$  exceeds the threshold fishing effort $E_{ext}$, then it results in the extinction of the whole species, i.e. $\lim_{t\to\infty} x(t)=0$ for any positive initial value. The mentioned asymptotic stability of the positive equilibrium, allows us to assume that, for a fixed fishing effort $E\in (0, E_{ext})$, the catch $C(E)$ is given by
$$ 
C(E)= E q \bar{x} = \frac{q^2}{a}E \left(\frac{a}{q}-E \right).
$$
In the paper we consistently use this simplification following many other authors.\\ 
\indent It is clear that maximal catch is obtained for the effort $E_{opt}$ given by
$$
E_{opt}=\mbox{argmax}\; C(E)=\frac{a}{2q}=\frac{E_{ext}}{2}.
$$
The equilibrium $\bar{x}_{opt}$ corresponding to the optimal effort $E_{opt}$ equals to $1/2$ (half of the capacity), therefore the maximal catch is $C(E_{opt})= E_{opt} q \bar{x}_{opt}=\frac{a}{4}$.\\

\subsection{MPA for single species }
Now let us pass to the MPA model given by \eqref{r1} that is related to \eqref{log1}. In order to find a steady-state of \eqref{r1}, we solve the following system
\begin{equation}\label{r2}
\left\{\begin{array}{l}
x_1\left(a(1-x_1/(1-R))-q E - mR\right) +m(1-R)x_2=0\\
x_2\left(a(1-x_2/R)-m(1-R)\right)+mRx_1=0,
\end{array}\right.
\end{equation}
We say that $(x_1,x_2)$ is a {\em positive solution} if $x_1>0$ and $x_2>0$. Observe also that if $(x_1, x_2)$ is a solution of \eqref{r2}, then $x_1=0$ if and only if $x_2=0$, there are no solutions with only one positive coordinate.\\
\indent Let us see how the existence of positive equilibria depends on $E$ and $R$. To this end in (\ref{r2}) set 
\begin{eqnarray}\label{A1}
& A=\displaystyle{\frac{1}{a}}(1-R)(a-qE-mR),\ \ B=\displaystyle{\frac{m}{a}}(1-R)^2,\\
& C =\displaystyle{\frac{R}{a}}(a-m(1-R)),\ \ D=\displaystyle{\frac{m}{a}R^2}
\end{eqnarray}
to obtain an equivalent form
\begin{equation}\label{r3}
\left\{\begin{array}{l}
\displaystyle{ x_2= \frac{1}{B}x_1(x_1-A)}\\
\ \ \\
\displaystyle{x_1=\frac{1}{D}x_2(x_2-C)}.
\end{array}\right.
\end{equation}
The existence of positive solution for \eqref{r3} is characterized in the following result.
\begin{prop}\label{alg-existence-prop}
Suppose that $B,\,D>0$. The system of equations \eqref{r3} has a unique positive solution if and only if
$A\geq 0$ or $C\geq 0$ or
\begin{equation}\label{A0A}
A<0 \mbox{ and } C<0 \mbox{ and } DB>AC.
\end{equation}
\end{prop}
\noindent Geometrical ideas of its proof can be find in \cite{Matsuda} (see  Figures 1-5).
\begin{cor}\label{C2}
Assume $a,m>0$ and $R\in (0,1)$. If  $m\leq a$ or $E<E_{ext}$ or
\[
a<m \ \mbox{ and } \ E>E_{ext} \ \mbox{ and }\  R> \left(1- \frac{E_{ext}}{E}\right)\left(1-\frac{a}{m}\right),
\]
then the system of equations \eqref{r2} has a unique positive solution $\left(\tilde{x}_1(E,R),\,\tilde{x}_2(E,R)\right)$.
\end{cor}
\noindent \begin{proof}  
Observe that $B, D>0$. If $A\geq 0$ or $C\geq 0$ then, due to Proposition \ref{alg-existence-prop}, \eqref{r3} and consequently (\ref{r2}) has a unique positive solution. So, it is sufficient to consider the case
$$
A<0 \mbox{ and } C<0.
$$
We have, by direct computation,
\begin{eqnarray}\label{A3}
DB-AC & = &  \frac{EqmR(1-R)}{a^2}\left(R - \left(1-\frac{E_{ext}}{E}\right)\cdot\left(1-\frac{a}{m}\right) \right).
\end{eqnarray}
The assumption $C<0$ implies that $m>a$. In the case $E\leq E_{ext}$, it is clear that $DB-AC>0$ for any $R\in (0,1)$. If $E>E_{ext}$, then
$DB-AC>0$ provided $R> (1-E_{ext}/E)(1-a/m)$. \end{proof}
\begin{remark}{\em
When implementing an MPA of size $R$ (usually when the fishing effort $E$ approaches $E_{ext}$), one may ask whether there is a security margin that prevents the extinction of the species even if the increase of $E$ will not be stopped. For a fixed $R\in (0,1)$, let $E_{ext}(R)$ be the infimum of the set of $E>0$ such that the species dies out, i.e. the only equilibrium of \eqref{r2} is the trivial one. Observe that, in view by Corollary \ref{C2}, we have  $E_{ext}(R)=+\infty$, for any $R\in (0,1)$, if only $m\leq a$. Clearly, it can be easily seen from the geometric analysis of \eqref{r3}, that, for large $E$, both $\tilde x_1(E,R)$ and $\tilde x_2(E,R)$ are small (hence in practice $E$ can not be increased arbitrarily). In the case $m>a$, we see that $C<0$ and, by Proposition \ref{alg-existence-prop}, we get the existence of a positive solution for \eqref{r3} if and only if $A\geq 0$ or $A<0$ and 
$DB-AC>0$. If $E > (a-mR)/q = E_{ext}\cdot (1-mR/a)$, then $A<0$ and the existence of positive solutions is then, in view of \eqref{A3}, equivalent to
$$
1 - \frac{R}{1-a/m} < \frac{E_{ext}}{E}.
$$
which yields $E_{ext}(R)=+\infty$ if $R \geq [1-a/m,1)$ and $E_{ext}(R)=E_{ext}/(1 - R/(1-a/m))$ for $R\in (1, 1-a/m)$.
}
\end{remark}

To verify the local asymptotic stability of the steady-state point $(\bar{x}_1,\bar{x}_2)$ of (\ref{r1}) observe that the linearization, i.e. the Jacobian of the vector field defined by the right-hand sides of the system equations, is given by
\[
\left[
\begin{array}{cc}
-\displaystyle{\frac{m(1-R)\bar{x}_2}{\bar{x}_1}}-\displaystyle{\frac{\bar{x}_1a}{1-R}}&m(1-R)\\
mR&-\displaystyle{\frac{mR\bar{x}_1}{\bar{x}_2}}-\displaystyle{\frac{a\bar{x}_2}{R}}
\end{array}
\right].
\]
By a simple calculation we obtain that the trace of the matrix is  negative and its  determinant is positive, which means that the $(\bar{x}_1,\bar{x}_2)$ is locally stable.

To verify the global asymptotic stability of the positive steady state point $(\bar{x}_1,\bar{x}_2)$ in the first quadrant $(0,\infty)\times(0,\infty)$ we construct a Lyapunov functional. First we change coordinates: $\frac{x_1}{1-R}\to x_1$, $\frac{x_2}{R}\to x_2$. In new coordinates system (\ref{r1}) has the following form:
\begin{equation}\label{r1A}
\left\{\begin{array}{l}
\dot x_1 =ax_1(1-x_1)-Eqx_1 +mR(x_2-x_1)\\
\dot x_2 =ax_2(1-x_2)-m(1-R)(x_2-x_1).
\end{array}\right.
\end{equation}
For the equilibrium points we clearly have
$$
a - Eq - mR= a\bar x_1 -mR\bar x_2/\bar x_1 , \ \ \ \ 
a -m(1-R) = a\bar x_2-m(1-R) \bar x_1/\bar x_2 ,
$$
which,  when put into \eqref{r1A},
yields the system (now without $E,q_x, q_y$ but with $\bar x_1, \bar x_2$)
\begin{equation}\label{r1B}  
(\dot x_1, \dot x_2)  = f (x_1, x_2),
\end{equation}
with $f= (f_1, f_2)$ given by
\begin{eqnarray*}
& & f(x_1,x_2) = ax_1(\bar{x}_1-x_1)+\frac{mR}{\bar{x}_1}(\bar{x}_1 x_2-\bar{x}_2x_1),\\
& & f_2 (x_1,x_2)=ax_2(\bar{x}_2-x_2)-\frac{m(1-R)}{\bar{x}_2}(\bar{x}_1x_2-\bar{x}_2x_1).
\end{eqnarray*}
\begin{prop}\label{1S-MPA-stablility}
For any $C_1, C_2>0$ such that $C_1R\bar{x}_2=C_2(1-R)\bar{x}_1$ the 
function $V:(0,+\infty) \times (0,+\infty) \to \R$ given by
$$
V(x_1,x_2)=C_1\varphi(x_1,\bar{x}_1)+C_2\varphi(x_2,\bar{x}_2),
$$ 
\begin{equation}\label{r1C}
\varphi(x,\bar{x})=x-\bar{x}-\bar{x}\ln\left(x/\bar{x}\right),
\end{equation}
is a Lyapunov functional for \eqref{r1B}, i.e. for any $(x_1, x_2) \in 
(0,+\infty)\times (0,+\infty) \setminus \{ (\bar x_1, \bar x_2)\}$,
\begin{eqnarray*}
V(x_1,x_2)>V(\bar{x}_1,\bar{x}_2)=0, \ \ \ \ \ \nabla V(x_1,x_2) \cdot f(x_1,x_2) < - \rho |(x_1, x_2)-(\bar x_1, \bar x_2)|^2 
\end{eqnarray*}
for some fixed $\rho>0$. Hence the equilibrium $(\bar x_1, \bar x_2)$ is asymptotically globally stable.
\end{prop}
\begin{proof} Indeed, by direct calculations we get
$$
\nabla V(x_1, x_2)\cdot f(x_1, x_2) =-C_1a(x_1-\bar{x}_1)^2-C_2a(x_2-\bar{x}_2)^2-\frac{C_1 R m a}{\bar{x_1}}\frac{1}{x_1\,x_2}(x_1\bar{x}_2-x_2\bar{x}_1)^2,
$$
which ends the proof.
\end{proof}

\subsection{Benefits of establishing an MPA in one species model}
The catch volume $C=C(E,R)$ is given by $$
C(E,R)=Eq\bar{x}_1,
$$
where $(\bar x_1, x_2)=(\bar{x}_1 (E,R),\,\bar{x}_2(E,R))$ is a positive solution to (\ref{r2}). If the system (\ref{r2}) have no positive solution then $C=0$.
We skip analytical formulas of the steady-state point $(\bar{x}_1,\bar{x}_2)$. 
Instead, we numerically analyze  several scenarios. By a scenario we understand the choice of the values of parameters in the model. The optimal size of MPA ($R_{opt}$) that maximizes the catch volume depends to the choice of $E$ - the fishing effort and $m$ - the mobility. We assume that the growth rate $a=1$ and the catchability $q=0.7$. We consider scenarios with four values of $E$ and three values of the mobiility $m$.\\
\indent We shall examine the model with the sustainable fishing effort $E_1=E_{opt}$, a moderate case when $E_2=(E_{opt}+E_{ext})/2$ (exceeding the optimal value but us still below the extinction one $E_{ext}$) and, regarding over-exploitation of many stocks, $E_3 = E_{ext}$, $E_4=2\,E_{ext}$. The mobility parameter comes from \cite{Takashina} where the value $m=0.1$ corresponds to a sedentary species, $m=1$ to a slowly migrating one and $m=10$ to a fast migrating. The differences between the values of mobility $10$ and $100$ are insignificant.\\
\indent In the considered example with no MPA, the maximal catch is obtained for $E=E_{opt}\approx 1.4286$ and $C(E_{opt})=0.25$. The maximal catch $C(E_{opt})$ will be our reference value for the catch in the model with MPA, namely we will normalize the catch by dividing it by $C(E_{opt})$. Let us also mention that the extinction fishing effort $E_{ext}\approx 2.8571$.
\begin{figure}[ht]
  \centering
  \includegraphics[keepaspectratio=true,width=16cm]{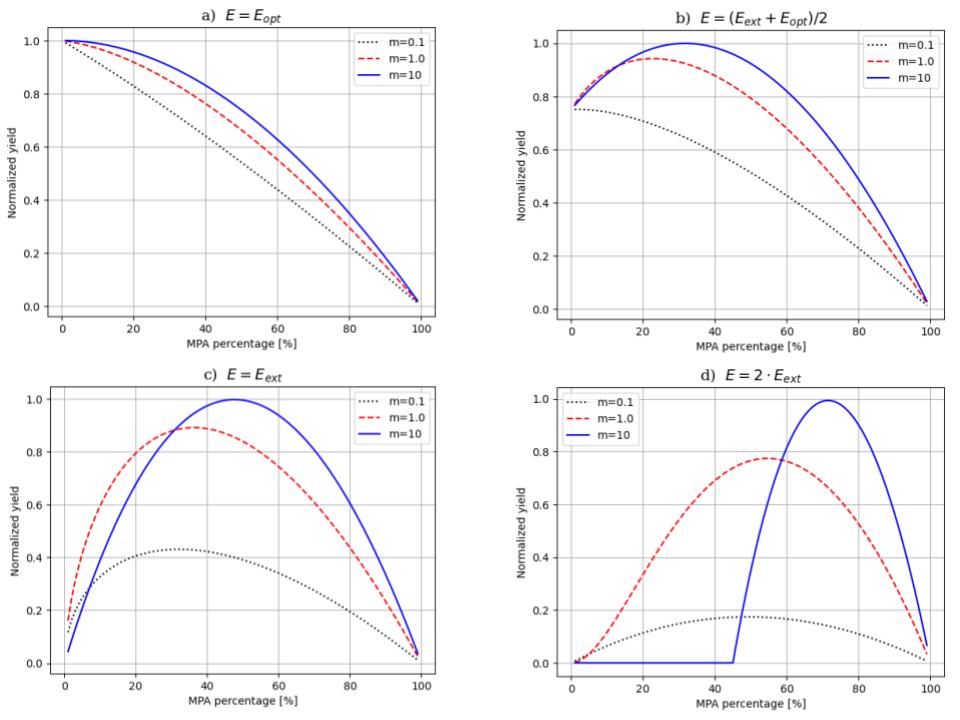}
  \caption{The normalized catch as a function of MPA size $R$ in one species model. The parameter values are $a=1$, $q=0.7$}. \label{Fig1}
\end{figure}
\indent For the sustainable fishing effort $E= E_{opt}$, establishing an MPA will result in a decrease of the catch independently of the mobility level -- see  Figure \ref{Fig1}a. A similar behavior we observe for $E<E_{opt}$. So, for $E<E_{opt}$, one does not expect any MPA benefits of in terms of catch.\\
\indent For a moderate fishing effort $E=(E_{opt}+E_{ext})/2$ and sedentary fish ($m=0.1$), introducing an MPA decreases the catch that is equal to $75\%$ of the optimal $C(E_{opt})$ for $R=0\%$. But for the higher mobilities $m=1$ and $m=10$, implementing an MPA will allow to get closer the optimal catch. For $m=1$ and $R=25\%$ the catch is $94\%$ of the maximal one $C(E_{{opt}})$ and for $m=10$ the optimal MPA size is $R=30\%$, for which the catch is almost $100\%$ of the maximal once -- see Figure \ref{Fig1}b.\\
\indent For the fishing effort $E=E_{ext}$ that normally (that is with no MPA) leads to the extinction of the stock of fish, the optimal MPA size and the catch strongly depend on the mobility -- see Figure \ref{Fig1}c. For sedentary fish ($m=0.1$) the optimal size of the MPA is $R=30\%$ and the corresponding catch amounts to $43\%$ of the maximal one. For slow mobility fish ($m=1$) optimal size of the MPA is $R=35\%$ and the corresponding catch is around $89\%$ of the maximal one. For the fast mobility ($m=10$), the optimal size of the MPA is $R=45\%$ and the catch is almost $100\%$ of the maximal one.\\

\newpage
\indent In the case of very high fishing effort $E=2\, E_{ext}$ and the slow and fast mobilities, the catch is zero if size of the MPA is less than $40\%$ and $50\%$, respectively. This observation is consistent with Corollary \ref{C2}. Nevertheless, for the fast mobility ($m=10$), if the MPA is large ($R=70\%$), then the catch exceeds $97\%$ of the maximal one. For the sedentary and slow mobilities, the optimal size of the MPA is $50\%$ and $55\%$, respectively, but the corresponding catch is below $17\%$ and $77\%$ of the maximal one, respectively.\\
\indent Summing up, in the considered one species model we observe yield benefits of establishing a marine reserve for over-exploited fisheries. The level of the yield benefits is larger for a mobile species. The conclusion is consistent with the other one species model analysis provided in \cite{Costello SR}, \cite{Matsuda}, \cite{Sanchirico}. Moreover, we observed that the total biomass $\bar{x}_1+\bar{x}_2$ at the steady-states is an increasing function of the marine reserve size $R$, also for sustainable fishing efforts. In all cases, implementing an MPA brings conservation benefits.

\section{Prey-predator model without MPA}
In the section we consider the  prey-predator model with fisheries ruled by the system of differential equations (\ref{pp1K}). Assuming that the capacity $K$ of the prey stock is equal to $1$ we obtain
\begin{equation}\label{pp1}
\left\{\begin{array}{l}
\dot x=ax(1-x)-bxy-Eq_x x,\\
\dot y=cxy-dy-Eq_y y,
\end{array}\right.
\end{equation}
where all other  parameters are like in (\ref{pp1K}).
Obviously, the system has up to three equilibrium points (they may possibly coincide)
\[
 (\bar{x},\bar{y}),\; (0,0), \; (\hat{x},0),
 \] 
where
\begin{equation}\label{pp1a}
 \bar{x}=\bar{x}(E)=\frac{d+Eq_y }{c}, \ \ \ \; \bar{y}=\bar{y}(E)=\frac{a(1-\bar{x})-Eq_x }{b},\; \ \ \ \hat{x}=\hat{x}(E)=\frac{a-Eq_x }{a}.
 \end{equation}
We are interested in non-negative equilibrium points that are globally stable in the first quadrant.
\begin{prop}\label{pp2}
Let 
$$
E_{y-ext}=\frac{a(c-d)}{a q_y+cq_x} \ \ \ \mbox{ and } \ \ \ E_{x-ext}=\frac{a}{q_x }.
$$
{\em (i)}  $0<E_{y-ext}<E_{x-ext}$ and
$$
\bar y(E)=\frac{aq_y+cq_x}{bc}(E_{y-ext}-E) \ \ \ \mbox{ and } \ \ \ 
\hat x(E)= 1- \frac{E}{E_{x-ext}}.
$$
{\em (ii)} If $0< E<E_{y-ext}$, then $(\bar x(E), \bar y(E))$ is a unique positive equilibrium of \eqref{pp1}
that is asymptotically stable in the first quadrant, i.e. $(0,+\infty)\times (0,+\infty)$
and $(\hat x(E),0)$ is an equilibrium with $\hat x(E)>0$.\\ 
{\em (iii)} If $E=E_{y-ext}$, then the prey is subject to extinction, i.e. $\bar y(E_{y-ext})=0$, and, moreover $\bar x(E_{y-ext}) = \hat x(E_{y-ext})$, which means that $(\hat x, 0)$ is then the only non-negative equilibrium of \eqref{pp1}. If $E=E_{x-ext}$, then $\hat x (E_{x-ext})=0$ and $(0,0)$ is the only equilibrium of \eqref{pp1}.\\
{\em (iv)} If $E_{y-ext}\leq E< E_{x-ext}$, then $0<\hat x(E)<\bar x(E)$ and $(\hat{x},0)$ is the only non-negative equilibrium point of \eqref{pp1} and it is globally asymptotically stable in the first quadrant.\\
{\em (v)} If $E\geq E_{x-ext}$, then $(0,0)$ is the unique equilibrium of \eqref{pp1} and it is globally asymptotically stable in the first quadrant.
 \end{prop}
\begin{proof}
(i) Observe that, since $c>d$, one has
$$
0<E_{y-ext} = \frac{a(1-d/c)}{a q_y/c+q_x} < \frac{a}{q_x}=E_{x-ext}.
$$
The formulas for $\bar x(E)$ and $\bar y(E)$ are immediate.\\
\indent (ii) By (i), for  $E\in (0,E_{y-ext})$, we get $\bar y(E)>0$ and 
$\hat x(E)> 1- E_{y-ext}/E_{x-ext}>0$.
The system \eqref{pp1} can be equivalently written (in a similar way to the passage from \eqref{r1A} to \eqref{r1B}) as
$(\dot x,\dot y) = g(x, y)$ 
with $g=(g_1, g_2)$ given by
$$
g_1(x,y) = -x(a(x-\bar{x})+b(y-\bar{y})) \ \ \mbox{ and } \ \ 
g_2(x,y) = cy(x-\bar{x}).
$$
We set $V:(0,+\infty)\times (0,+\infty)\to\R$ as  $V(x,y)=C_1\varphi(x,\bar{x})+C_2\varphi(y,\bar{y})$,
where $\varphi$ is given by (\ref{r1C}) and $C_1, C_2>0$ ara such that $C_1 b = C_2c$. It is clear that $\nabla V(x,y)\cdot g(x,y)=-a(x-\bar{x})^2$, which is negative outside the line $x=\bar{x}$ and the vector field $g$ is transversal to the line at $(x,y)\neq (\bar x, \bar y)$. Hence $(\bar x, \bar y)$ is asymptotically stable in the first quadrant.\\
\indent (iii) can be verified by a straight forward computation.\\
\indent (iv) Let us fix $E \in (E_{y-ext}, E_{x-ext})$ and set 
$$
O_1:=\{(x,y): 0<x< \bar x(E),\, y>0\}, \ \ \ O_2:=\{(x,y): x>\bar x(E),\,y>0\}.
$$
Since $E\geq E_{y-ext}$, we simply have $\bar y = \bar y(E) \leq 0$. Then
\begin{equation}\label{O2-prop}
 g_1(x,y)\leq x (b\bar y - by) \leq  -bxy<0 \ \ \ \mbox{ and } \ \  g_2(x,y)\geq 0 \mbox{ for all }
 (x,y)\in O_2.
\end{equation}
We claim that a trajectory $(x(\cdot),y(\cdot))$ starting from a  point $(x_0,y_0)\in O_2$ passes to $O_1$ in a finite time. Indeed, suppose to the contrary that $(x(t),y(t))\in O_2$ for all $t>0$.
Then, by \eqref{O2-prop}, $\dot y(t)\geq 0$, i.e. $y(t)\geq y_0$ for all $t>0$, and
$$
\dot x(t) = g_1(x,y)\leq -bx(t)y(t) \leq -by_0 x(t) \mbox { for all } t>0, 
$$
which clearly implies that $x(t)\leq x_0  \exp(-by_0t)\rightarrow 0$ as $t\to +\infty$, a contradiction.\\
\indent Observe also that if $x=\bar x$ and $y>0$, then $g_1(x,y)\leq b\bar x\bar y <0$, which means that no trajectory can leave $O_1$ crossing the line $x=\bar x$. It is also clear that no trajectory can leave $O_1$ through the line $x=0$ or $y=0$, since it would contradict the backward in time uniqueness of solutions for the system (\ref{pp1}). Hence, no trajectory leaves $O_1$.\\ 
\indent Furthermore, since $g_1(x,y)=-x(ax+by-a+E q_x)$, we see that $g_1$ is non-negative in the triangle $O_3$ with vertices $(0,0)$, $(\hat{x},0)$ and $(0,(a-Eq_x)/b)$ and negative outside this set, i.e. in $O_1\setminus O_3$ (note that $a-E q_x>0$ due to the assumption $E<E_{x-ext}$). It is also clear that $g_2(x,y)<0$ for all $(x,y)\in O_1$. Therefore no trajectory leaves $O_3$.\\
\indent Now take a trajectory 
$(x(\cdot), y(\cdot))$ in $O_1$. The function $y$ is obviously decreasing and converging to some $\tilde y\geq 0$. If the trajectory starts in $O_3$ then it stays in $O_3$ and $x$, as a non-decreasing function, converges to some $\tilde x>0$. Hence $g(\tilde x, \tilde y)=0$, which implies $(\tilde x, \tilde y)=(\hat x,0)$. If $(x(\cdot),y(\cdot))$ starts in $O_1\setminus O_3$, then $x$ is decreasing, which means that either the trajectory stays in $O_1\setminus O_3$ or enters $O_3$, in both cases we get  $(x(t),y(t)) \to (\hat x,0)$ as $t\to+\infty$. It follows from these considerations that $(\hat x,0)$ is globally asymptotically stable.\\
\indent (v) We adopt the same arguments as in the proof of (iv) to show that any trajectory starting from a point in $O_2$ pass to $O_1$ in a finite time. If  $(x,y)\in O_1$ then $g_1(x,y)<0$ and $g_2(x,y)<0$. Thus, any trajectory $(x(t),y(t))$ of (\ref{pp1}) is convergent to $(0,0)$ as $t\to \infty$.
\end{proof}

The predator catch at the equilibrium is given by
\begin{equation}\label{pp3a}
C_{Pred}(E)=Eq_y\bar{y}(E)=\left\{ 
\begin{array}{ccl}
q_yE(a(c-d)-(aq_y+cq_x)E)/bc&\mbox{ if} & E\leq E_{y-ext},\\
0&\mbox{ if} &E>E_{y-ext}.
\end{array}
\right.
\end{equation}
From the food security point of view the payoff function $C(E)$  that we shall maximize is the total catch. The total catch we define as the sum of the preys catch  and the  predator catch. By Proposition \ref{pp2} we have
\begin{equation}\label{pp4}
C(E)=\left\{\begin{array}{lll}
Eq_x \bar{x}(E)+Eq_y \bar{y}(E)&\mbox{ if }&0< E\leq E_{y-ext},\\
Eq_x \hat{x}(E)&\mbox{ if }& E_{y-ext}<E\leq E_{x-ext},\\
0&\mbox{ if }& E_{x-ext}<E.
\end{array}\right.
\end{equation}
In view of Proposition \ref{pp2} (iii), the functional $C$ is continuous. In general, it is not differentiable at $E=E_{y-ext}$. By (\ref{pp1a}), we have
\begin{equation}\label{pp5}
C(E)=\left\{\begin{array}{lll}
E(G-ED))&\mbox{ if }&0< E\leq E_{y-ext},\\
Eq_x (1-E q_x/a)&\mbox{ if }& E_{y-ext}<E\leq E_{x-ext},\\
0&\mbox{ if }& E_{x-ext}<E.
\end{array}\right.
\end{equation}
where $D=q_y \big(a q_y - (b-c) q_x\big)/bc$
and $G=aq_y (c-d)/bc+dq_x /c$. Clearly, $G>0$ and $D$ may have different signs.\\
\indent If $q_x / q_y \geq a/(b-c)$, then $D\leq 0$ and the function $C(E)$ is increasing on the interval $(0,E_{y-ext})$. So the maximum of $C$ is attained in the interval $(E_{y-ext},E_{x-ext})$.\\
\indent If $q_x/q_y < a/(b-c)$, then $D>0$ and the function $E\mapsto E(G-ED)$ has its maximum at $E_1=G/2D$. We have
$$
E_1<E_{y-ext} \ \ \ \mbox{ if and only if }  \ \ \ 
A (q_x/q_y)^2 + B(q_x/q_y) -a<0,
$$
with $A=bcd/a(c-d)$ and $B= b-c + bc/(c-d)$.
Hence, $E_1<E_{y-ext}$  if and only if $q_x/q_y <Q$ where $Q=(\sqrt{B^2+4Aa}-B)/2A$. Let us also observe that, by a direct computation, we have
\begin{equation}\label{Q-prop}
Q<\frac{a}{B} <\frac{a}{b-c}.
\end{equation}
\indent If $q_x/q_y <Q$ then $C$ has a local maximum at $E=E_1<E_{y-ext}$ and if $Q< q_x/q_y < a/(b-c)$, then the function $C$ is increasing on the interval $(0,E_{y-ext})$.\\ 
\indent The function $E\mapsto Eq_x(1-Eq_x/a)$ has its maximum at $E=E_2=a/2q_x=E_{x-ext}/2$. We easily see that
\[
E_2>E_{y-ext} \ \ \ \mbox{ if and only if  } \  \ \  aq_y >(c-2d)q_x.
\]
If $q_x/q_y<Q$, then $E_2 > E_{y-ext}$. Indeed, if $c-2d\leq 0$ then trivially
$(c-2d)q_x<a q_y$, which implies $E_2>E_{y-ext}$; if $c-2d>0$, then
$B>c-2d$ and, in view of \eqref{Q-prop}, one has
$$
\frac{q_x}{q_y} < Q < \frac{a}{B} < \frac{a}{c-2d},
$$
i.e. $(c-2d)q_x<aq_y$, which gives $E_2>E_{y-ext}$.\\
\indent  Hence, if $q_x/q_y<Q$, then $C$
has two maxima at $E_1$ and $E_2$. We claim that then
$C(E_1)<C(E_2)$, i.e. $G^2/4D < a/4$.    
To show it note first that $G^2/D = u(q_x/q_y)$ where $u:(0,a/(b-c))\to \R$ is given by
$$
u(q)=\frac{\big(a(c-d)+bdq\big)^2}{bc\big(a-(b-c)q\big)}. 
$$
Since $u$ is increasing, we have, in view of \eqref{Q-prop},
$$
u(q)<u(Q)<u(a/B).
$$
Further
\begin{eqnarray*}\label{pp5a}
u(a/B)& =&  \frac{a}{bc} \cdot \frac{\big((c-d) - bd/B\big)^2}{1-(b-c)/B} = \frac{a}{Bbc}\cdot \frac{\big(B(c-d)-bd \big)^2}{B-(b-c)}= a K
\end{eqnarray*}
where 
$$
K=\frac{(2b-c)^2(c-d)^3}{(bc)^2 \big( b-c + bc/(c-d))\big)}.
$$
Note that $d<c$ implies $b-c+bc/(c-d) > 2b-c$ and $(c-d)^3< c^3$, which yields
\begin{eqnarray*}
K & < & \frac{c(2b-c)}{b^2} \leq 1.
\end{eqnarray*}
This proves that $u(q)<u(a/B)<a$, i.e. $C(E_1)<C(E_2)$.

We can sum up the above considerations in the following proposition.
\begin{prop}\label{pp6}
{\em (i)} If $q_x/q_y \geq Q$, then the function $C$ is increasing on the interval $(0,E_{y-ext})$ and it has a maximum in the interval $[E_{y-ext},E_{x-ext})$.\\
{\em (ii)} If $q_x/q_y <Q$, then the function $C$ has two local maxima. The first one $E_1=G/2D$ in the interval $(0,E_{y-ext})$, the second one $E_2=E_{x-ext}/2$ in the interval $(E_{y-ext},E_{x-ext})$ and   $C(E_1)<C(E_2)$.
 \end{prop}
Thus, it appears that whatever are the parameters in this prey-predator model, the total catch is maximal for the fishing effort that (sometimes significantly) exceeds the prey extinction threshold $E_{y-ext}$. It means that such systems are naturally susceptible to over-exploitation meaning the risk of prey extinction. In such fisheries implementing a protection area may be considered as a preventive measure.



\section{The prey-predator model with MPA}
We consider the prey-predator model with MPA, that is \eqref{m1}, with the normalized capacity $K=1$,
 \begin{equation}\label{m1A}
  \left\{
        \begin{array}{l}
\dot x_1= a  x_1(1-x_1/(1-R))- b x_1 y_1 /(1-R)-Eq_x  x_1+m_x ((1-R)x_2-Rx_1)  \\
\dot x_2= a  x_2(1-x_2/R)- bx_2y_2/R -m_x ((1-R)x_2-Rx_1) \\
  \dot y_1= c x_1 y_1/(1-R)- d  y_1-Eq_y y_1+m_y  ((1-R)y_2-Ry_1)\\
 \dot y_2= c x_2y_2/R- d  y_2-m_y  ((1-R)y_2-Ry_1)
        \end{array} \right.
 \end{equation}
where all the coefficients are like in \eqref{m1}. We shall prove the uniqueness of positive steady state solutions for the system (even for a broader class of problems). Next we show the existence of the positive steady state. We also discuss the asymptotic stability of the positive steady states.

We change variables in (\ref{m1A})
\[
\frac{x_1}{1-R}\to x_1,\;\; \frac{y_1}{1-R}\to y_1,\;\;\frac{x_2}{R}\to x_2,\;\;\frac{y_2}{R}\to y_2
\]
and obtain
\begin{equation}\label{m2}
\left\{\begin{array}{l}
\dot x_1= a  x_1(1-x_1)- b  x_1y_1-Eq_x  x_1+m_x R(x_2-x_1)  \\
\dot x_2= a  x_2(1-x_2)- b  x_2 y_2-m_x (1-R)(x_2-x_1)  \\
\dot y_1= c  x_1 y_1- d  y_1-Eq_y y_1+m_y  R(y_2-y_1)\\
\dot y_2= c  x_2 y_2- d  y_2-m_y  (1-R)(y_2-y_1).
\end{array} \right.
 \end{equation}
\subsection{Uniqueness  of positive steady states}
Let us consider a generalization of \eqref{m2} given by
\begin{equation}\label{system-2sp-MPA}
\left\{ \begin{array}{l}
\dot x_1 = a x_1 \big (1\!-\!\alpha(x_1)\big) \!-\!\beta (y_1) x_1 \!-\! E q_x x_1 \!+\! m_{x}R \!\cdot\! \big(x_2\!-\!x_1\big),\\
\dot x_2 = a x_2 \big (1- \alpha(x_2)\big)  -\beta (y_2) x_2 - m_{x}(1-R) \cdot \big(x_2-x_1\big)\\
\dot y_1 = \gamma (x_1) y_1-d y_1- E q_y y_1 + m_{y}R \cdot (y_2-y_1),\\
\dot y_2 = \gamma (x_2) y_2- d y_2 - m_{y}(1-R) \cdot (y_2-y_1),
\end{array} \right.
\end{equation}
where $\alpha, \beta, \gamma:\R\to\R$ are continuous increasing bijections with
$\alpha(0)=\beta(0)=\gamma(0)=0$. It is clear that it has the form of the following equations
\begin{equation}\label{system_0}
 \left\{ \begin{array}{l}
\dot x_1 = x_1 \big( \theta_1 - \alpha(x_1) - \beta(y_1) \big) + \mu_x x_2 = 0\\
\dot x_2 = x_2 \big( \theta_2 - \alpha(x_2) - \beta(y_2) \big) + \rho \ mu_x x_1 = 0\\
\dot y_1 = y_1 \big(\gamma(x_1) -  \nu_1 \big)   + \mu_y y_2 = 0\\
\dot y_2 = y_2 \big(\gamma(x_2) - \nu_2 \big) + \rho \mu_y y_1  = 0.
\end{array} \right.
\end{equation}
with $\alpha$ multiplied by $a$ but denoted again by $\alpha$ and
\begin{equation}\label{coeff1}
\theta_1 = a - Eq_x - m_{x}R, \ \ \theta_2 = a -m_{x}(1-R),  \ \ \rho = \frac{1-R}{R},
\end{equation}
\begin{equation}\label{coeff2}
\mu_x = m_{x} R, \ \ \mu_y = m_{y}R, .
\end{equation}
\begin{equation}\label{coeff3}
\nu_1 = d + Eq_y + m_{y}R, \ \ \nu_2 = d + m_{y}\rho R.
\end{equation}
\noindent  We consider positive steady states of the system, i.e. positive solutions $(x_1, x_2, y_1, y_2)$ of the following system
\begin{equation}\label{system_1_normalized}
 \left\{ \begin{array}{l}
x_1 \big( \theta_1 - \alpha(x_1) - \beta(y_1) \big) + \mu_x x_2 = 0\\
x_2 \big( \theta_2 - \alpha(x_2) - \beta(y_2) \big) + \rho \mu_x x_1 = 0\\
y_1 \big(\gamma(x_1) -  \nu_1 \big)   + \mu_y y_2 = 0\\
y_2 \big(\gamma(x_2) - \nu_2 \big) + \rho \mu_y y_1  = 0.
\end{array} \right.
\end{equation}
\indent Below we show that positive equilibria, if exist, are unique. Its proof gives us also some hints how to search for positive equilibria in the general setting.
\begin{thm}\label{necessary-condition}
Suppose that $\alpha, \beta, \gamma:\R\to\R$ are continuous increasing bijections with
$\alpha(0)=\beta(0)=\gamma(0)=0$, $\mu_x, \mu_y, \rho>0$ and $\theta_1,\theta_2\in\R$. Then the system  \eqref{system_1_normalized} has at most one positive solution.
\end{thm}
\begin{proof} Let $(x_1, x_2, y_1, y_2)$ be a positive solution to \eqref{system_1_normalized}. 
Since $(y_1, y_2)$ is a nontrivial solution of the system consisting of the third and forth equations in \eqref{system_1_normalized} and $\mu_y, \rho>0$, we get $\gamma(x_1)-\nu_1 <0$, $\gamma(x_2) - \nu_2<0$ and
\begin{equation}\label{x-1-x-2-relation}
(\gamma(x_1) - \nu_1)(\gamma(x_2)- \nu_2) = \rho \mu_y^2.
\end{equation}
Therefore $\nu_1 \nu_2 > \big( \nu_1-\gamma(x_1) \big) \big( \nu_2- \gamma(x_2) \big) = \rho  \mu_y^2$.
We set
$$
x^{*}_1= \gamma^{-1}\left( \nu_1 - \rho \mu_y^2/ \nu_2 \right) \ \mbox{ and }  x^{*}_2 = \gamma^{-1}
\left( \nu_2 - \rho \mu_y^2/ \nu_1 \right);
$$
and define a function $g:[0, x^{*}_1] \to [0,  x^{*}_2]$  by
$$
g(s)= \gamma^{-1} \left( \nu_2 - \frac{\rho \mu_y^{2}}{\nu_1 - \gamma(s)} \right), \, s\in [0,x^{*}_1].
$$
 Clearly \eqref{x-1-x-2-relation} implies $x^{*}_2 = g(x^{*}_1)$. It is easy to verify that $g$ is well-defined, decreasing and $g(0)= x^{*}_2$
and $g( x^{*}_1)=0$.\\
\indent It follows from the first and second equations that
\begin{equation}\label{formulas-for_y1-y2}
y_1 = \beta^{-1}\left( \theta_1-\alpha(x_1) +\mu_x \frac{x_2}{x_1} \right) \ \mbox{ and } \ y_2 = \beta^{-1}\left( \theta_2 - \alpha(x_2) + \rho \mu_x \frac{x_1} {x_2} \right)
\end{equation}
and substituting it into the third equation we get
\begin{equation}\label{x1-key-equation}
\big(\nu_1-\gamma(x_1)\big)\beta^{-1}\left( \theta_1-\alpha(x_1) +\frac{\mu_x g(x_1)}{x_1} \right) = \mu_y \beta^{-1}\left( \theta_2 - \alpha(g(x_1)) + \frac{\rho \mu_x x_1} {g(x_1)} \right)
\end{equation}
From the formula for $y_1$ in \eqref{formulas-for_y1-y2} and the fact that $y_1>0$, we get $\theta_1-\alpha(x_1)+\mu_x x_2/x_1>0$.\\
\indent Define $h_1: (0, x^{*}_1 ] \to \R$ and $h_2:[0, x^{*}_1)\to\R$ by the formulas
and
\begin{eqnarray*}    
h_1 (s) & = & \big(\nu_1-\gamma(s)\big)\beta^{-1}\left( \theta_1-\alpha(s) +\frac{\mu_x g(s)}{s} \right), \, s\in (0, x^{*}_1],\\
h_2 (s) & = & \mu_y \beta^{-1}\left( \theta_2 - \alpha(g(s)) + \frac{\rho \mu_x s} {g(s)} \right), \, s\in [0,  x^{*}_1).
\end{eqnarray*}
Clearly \eqref{x1-key-equation} has the form $h_1(x_1)=h_2(x_1)$. Note that
$$
\lim_{s\to 0^+} h_1(s)=+\infty.
$$
Since $\alpha$ is increasing and $g$ is decreasing,  the function $s\mapsto \theta_1 - \alpha(s)+\mu_x g(s)/s$ is decreasing on the interval $(0, x^{**})$ where
$$
x^{**} = \sup \big\{s\in (0, x^{*}_1] \mid \theta_1 - \alpha(s)+\mu_x g(s)/s >0\big\}
$$
We also have that $\nu_1 - \gamma(s) > \nu_1 -\gamma( x^{*}_1)>0$ for all $s\in (0, x^{*}_{1})$, which shows that
$h_1$ is decreasing on the interval $(0, x^{**})$. In the same way we easily see that $h_2$ is increasing and 
$$
\lim_{s\to (x^{*}_1)^-} h_2(s)= +\infty.
$$
\indent Observe that, since $y_1>0$, we have $x_1\in (0, x^{**})$. Therefore $x_1$ is a unique solution of \eqref{x1-key-equation} and, consequently, $(x_1, x_2, y_1, y_2)$ is the unique positive solution for \eqref{system-2sp-MPA}. \end{proof}

\begin{remark}\label{exitstence-rem} {\em
(a) The condition $\nu_1\nu_2>\rho \mu_y^2$, which is necessary for the existence of the positive steady states of \eqref{system-2sp-MPA}, is satisfied for the coefficients defined by \eqref{coeff1} and \eqref{coeff2}.\\
\indent (b) Observe that actually the proof of Theorem \eqref{necessary-condition} provides a sort of method for finding positive steady states.  Indeed, if $x^{**} = x^*_1$, then there exists the unique $x_1\in (0, x^*_1)$ such that $h_1(x_1)=h_2(x_1)$ and, with 
$x_2=g(x_1)$ and $y_1$, $y_2$ given by \eqref{formulas-for_y1-y2}, $(x_1, x_2, y_1, y_2)$ is the unique positive solution of \eqref{system_1_normalized}. In the case $x^{**}<\ x^{*}_{1}$ and
$\theta_2 - \alpha(g( x^{**})) + \rho \mu_x  x^{**}/g( x^{**})>0$,
then there is a unique $x_1\in (0, x^{**})$ such that $h_1(x_1)=h_2(x_1)$, which yields also the existence of a positive solution of \eqref{system_1_normalized}. 
} \end{remark}
\begin{remark}\label{extinction-order} {\em
Let $(x_1, x_2, y_1, y_2)$  be a solution of the system \eqref{system_1_normalized} with the parameters like in Theorem \ref{necessary-condition}. It follows from the first equation that if $x_1=0$, then $x_2=0$ and, from the second we see that, if $x_2=0$, then $x_1=0$. Moreover, $x_1=x_2=0$ entails $y_1=y_2=0$. Indeed, from the third and fourth equations, we have $-\nu_1 y_1 + \mu_y y_2=0$ and $-\nu_2 y_2 + \rho \mu_y y_1=0$, which implies $y_1=y_2=0$ due to the fact that $\nu_1 \nu_2 < \rho \mu_y^2$.} \end{remark}


\subsection{Existence of positive steady states}
If $m_x =m_y  =0$ then system (\ref{m2}) splits into two independent systems. The first one describes fish densities $x_1(t),\,y_1(t)$ outside the MPA and the second one describes  fish densities $x_2(t),\,y_2(t)$ inside the MPA. By Proposition \ref{pp2}, the steady state $(\bar{x}_1, \bar{y}_1)$ of the system describing fish densities  outside the MPA is given by
 \begin{equation}\label{pp2A}
\bar{x}_1=\frac{d+Eq_y}{c},\;\;\bar{y}_1=\frac{a(c-d)-E(aq_y+cq_x)}{bc}
 \end{equation}
and the steady state $(\bar{x}_2,\bar{y}_2)$ of the system describing fish densities inside the MPA is given by
 \begin{equation}\label{pp2B}
\bar{x}_2=\frac{d}{c},\;\;\bar{y}_2 =\frac{a(c-d)}{bc}.
 \end{equation}
Obviously 
\[
\bar{x}_1>\bar{x}_2>0,\;\;\bar{y}_1<\bar{y}_2, \;\;\;\bar{y}_2>0.
\]
Hence, the system (\ref{m2}) is equivalent to
\begin{equation}\label{m3}
\left\{ \begin{array}{l}
\dot x_1 = x_1(- a (x_1-\bar{x}_1)- b (y_1-\bar{y}_1))+m_x R(x_2-x_1)  \\
\dot  x_2=x_2(- a  (x_2-\bar{x}_2)- b  (y_2-\bar{y}_2))-m_x (1-R)(x_2-x_1)  \\
\dot  y_1=cy_1   (x_1-\bar{x}_1)+m_y   R(y_2-y_1)\\
\dot  y_2=cy_2  (x_2 -\bar{x}_2)-m_y  (1-R)(y_2-y_1).
\end{array} \right.
 \end{equation}
If we put
  \[
  A=\frac{ b }{ a }, \ \ B=\frac{m_y  R}{ c },\ \ C=\frac{m_x R}{ a },\ \ D=\frac{1-R}{R},
  \]
then the system 
\begin{equation}\label{m4}
\left\{
\begin{array}{l}
x_1(- a (x_1-\bar{x}_1)- b (y_1-\bar{y}_1))+m_x R(x_2-x_1)=0  \\
x_2(- a  (x_2-\bar{x}_2)- b  (y_2-\bar{y}_2))-m_x (1-R)(x_2-x_1)=0  \\
cy_1   (x_1-\bar{x}_1)+m_y   R(y_2-y_1)=0\\
cy_2   (x_2 -\bar{x}_2)-m_y  (1-R)(y_2-y_1)=0
\end{array} \right.
\end{equation}
can be equivalently written as
\begin{equation}\label{m4A}
\frac{1}{C}x_1(x_1-(\bar{x}_1+A\bar{y}_1-Ay_1-C))=x_2,
\end{equation}
\begin{equation}\label{m4B}
\frac{1}{CD}x_2(x_2-(\bar{x}_2+A\bar{y}_2-Ay_2-CD))=x_1,
\end{equation}
\begin{equation}\label{m4C}
y_1(x_1-\bar{x}_1-B)+y_2B=0,
\end{equation}
\begin{equation}\label{m4D}
 y_2(x_2-\bar{x}_2-BD)+y_1BD=0.
\end{equation}
Note that $A$, $B$, $C$, $D$, $\bar x_2$, $\bar y_2$ do not depend on $E$ and only $\bar x_1$ and $\bar y_1$ do. By a geometric analysis of the above system we shall show the following existence result. 
\begin{thm}\label{EQ1.1}
 Suppose that  $a>b>c>d>0$, $m_x ,\,m_y  >0$ and $\bar{x}_1,\,\bar{y}_1,\,\bar{x}_2,\,\bar{y}_2$ are given by \eqref{pp2A}-\eqref{pp2B}.  If $R\in(0,1)$, then there exists $E_{y-ext}(R)>E_{y-ext}$ such that, for each $E\in (0, E_{y-ext}(R))$, the system of equations \eqref{m4} has a unique positive solution $(\tilde{x}_1,\tilde{x}_2,\tilde{y}_1,\tilde{y}_2)$ and if $E>E_{y-ext}(R)$ then the system \eqref{m4} has no positive solutions.
 Moreover, 
 we have $\tilde{y}_1<\tilde{y}_2$, $\tilde{x}_2>\bar{x}_2$ and $\tilde{x}_1<\bar{x}_1$.
 \end{thm}
\noindent In the proof we use an observation that the system \eqref{m4C}-\eqref{m4D} of linear equations with parameters $x_1$, $x_2$ has a positive solution $(y_1,y_2)$ ($y_1>0$, $y_2>0$) if and only if
\begin{equation}\label{m5C}
(x_1-\bar{x}_1-B)(x_2-\bar{x}_2-BD)=B^2D ,\ \ \ 0\leq x_1<\bar{x}_1+B,\ \ \ 0\leq x_2<\bar{x}_2+BD.
\end{equation}    
Then we shall be looking for positive $y_1$ and $y_2$ such that the parabolas given by  \eqref{m4A} and \eqref{m4B} intersect on the hyperbola given by \eqref{m5C}. To this end we shall need the following lemmata.
\begin{Lemma}\label{L1}
The system of equations
\begin{equation}\label{R1}
x_2=\frac{1}{C}x_1((x_1-\bar{x}_1)+C),
\end{equation}
\begin{equation}\label{R2}
x_1=\frac{1}{CD}x_2((x_2-\bar{x}_2)+CD),
\end{equation}
has a unique solution $(x^{*}_{1,0}, x^{*}_{2,0})\in(0,\infty)\times (0,+\infty)$ and 
$$
(x^{*}_{1,0}-\bar{x}_1-B)(x^{*}_{2,0}-\bar{x}_2-BD)<B^2D,
$$
which means that $(x^{*}_{1,0}, x^{*}_{2,0})$ lies above the hyperbola given by \eqref{m5C}, i.e. in the unbounded component of the positive quadrant that is cut out by the hyperbola. 
\end{Lemma}
\begin{proof} We have 
$x^{*}_{1,0}-\bar{x}_1=C(x^*_{2,0}/x^{*}_{1,0}-1)$ and $x^{*}_{2,0}-\bar{x}_2=CD(x^{*}_{1,0}/x^{*}_{2,0}-1)$, which gives
\begin{eqnarray*}
& & (x^{*}_{1,0}-\bar{x}_1-B)(x^{*}_{2,0}-\bar{x}_2-BD)-B^2D\\
& &  \ \ \ \ \ = \left(C (x^*_{2,0}/x^{*}_{1,0}-1)-B\right)\left(CD (x^{*}_{1,0}/x^{*}_{2,0}-1)-BD\right)-B^2D \\
& & \ \ \ \ \ = D(C^2+BC)\left(2-x^*_{2,0}/x^{*}_{1,0}-x^{*}_{1,0}/x^{*}_{2,0}\right)\leq 0,
\end{eqnarray*}
which ends the proof.
\end{proof}
\begin{Lemma}\label{L2}
If $a>0$ and $p_1<p_2$ then the positive part of the right arm of the parabola $x_2=ax_1(x_1-p_1)$ is situated to the right of the positive part of the right arm of the parabola $x_2=ax_1(x_1-p_2)$, i.e. if  $az_1(z_1-p_1)=az_2(z_2-p_2)>0$,  $z_1>0$,  $z_2>0$ then $z_1<z_2$.
\end{Lemma}
\begin{proof} Suppose to the contrary that $z_1\geq z_2$. If $z_1=z=z_2$ then $p_1=p_1 $, a contradiction. If $z_1>z_2$ then $z_1-p_1>z_2-p_2$. Moreover $z_2-p_2>0$ as $z_2>0$ and $z_2(z_2-p_2)>0$. Thus $z_1(z_1-p_1)>z_2(z_2-p_2)$, a contradiction.
\end{proof}
\begin{Lemma}\label{L3}
Let $r,q, p>0$, $p'>p$ and $z_1\in (0,p)$, $z'_1\in (0,p')$ and 
$z_2, z'_2\in (0,q)$ be such that
$$
(z_1-p)(z_2-q)=r \ \ \ \mbox{ and } (z'_1-p')(z'_2-q) = r. 
$$
If $z_1/z_2 = z'_1/z'_2$ then $z_1<z'_1$ and $z_1<z'_2$.
\end{Lemma}
\begin{proof} If we suppose to the contrary then $z_1\geq z'_1$ and $z_2\geq z'_2$. This yields
$$
r =(p-z_1)(q-z_2) < (p'-z'_1)(q-z'_2)<r, 
$$
a contradiction.
\end{proof}
\begin{proof}[Proof of Theorem \ref{EQ1.1}] 
Set 
$$
H=\{(x_1,x_2)\in \R^2\mid\; \eqref{m5C} \mbox{ holds true}\}
$$
and note that the set $H$ is the arc of the hyperbola given by \eqref{m5C} in the positive quadrant $(0,+\infty)\times (0,+\infty)$ with ends points $(0,\bar x^{*}_{2})$ and $(x^{*}_{1},0)$,  where
\[
x^{*}_{1}=\bar{x}_1+B-\frac{B^2D}{\bar{x}_2+BD},\ \ \ x^{*}_{2}=\bar{x}_2+BD-\frac{B^2D}{\bar{x}_1+B}.
\]
It is easily seen that $x^{*}_1>\bar x_1$ and $x^{*}_2> \bar x_2$. For any $y_1\in\R$ we define the set
$$
P_1(y_1)=\{(x_1,x_2)\mid (\ref{m4A}) \mbox{ holds true}\}.
$$
$P_1(y_1)$ is a parabola with vertical axis. We shall consider two cases
$$
\mbox{ (A)}  \ \ \bar{x}_1+A\bar{y}_1-C\leq  x^{*}_{1} \ \ \ \mbox{ and } \ \  \ \mbox{ (B)} \ \    \bar{x}_1+A\bar{y}_1-C>  x^{*}_{1}.
$$
In the case (A), the parabola $P_1(0)$ intersects with $H$ in one point that we denote by $(b_1,b_2)$.
If (B) holds, then we put $(b_1,b_2)=( x^{*}_{1},0)$. Let $H_1=\{(x_1,x_2)\in H:\;x_1\in(0,b_1)\}$ and note that
\begin{equation}\label{m5D}
\mbox{ for any } (x_1,x_2)\in H_1  \mbox{ there exists a unique } y_1>0 \mbox{ such that } \mbox{(\ref{m4A}) holds true.}
\end{equation}
The function $\hat y_1:(0,b_1)\to(0,\infty)$ given by (\ref{m5D}) is continuous, decreasing and
\[
\lim_{x_1\to0^+} \hat y_1(x_1)=+\infty,
\]
\[\lim_{x_1\to b_1^-} \hat y_1(x_1)=\left\{\begin{array}{ll}
0& \mbox{ if (A) holds, } \\
y_b &\mbox{ if (B) holds, }
\end{array}\right.\]
where $y_b>0$ solves
\[
\bar{x}_1+A\bar{y}_1-Ay_b-C= x^{*}_{1}.
\]
\indent For any fixed $y_2\in \R$ define 
$$ 
P_2(y_2)=\{(x_1,x_2)\in \R^2\mid (\ref{m4B}) \mbox{ holds true}\}.
$$ 
It is a parabola with horizontal axis. As before We shall consider two cases
$$
(C) \ \ \bar{x}_2+A\bar{y}_2-CD\leq  x^{*}_{2}  \ \ \ \mbox{ and } \ \  \ (D)  \ \ \bar{x}_2+A\bar{y}_2-CD>  x^{*}_{2}.
$$
In the case $(C)$, the parabola $P_2(0)$ intersects the set $H$ in one point that we denote by $(a_1,a_2)$. If in the case $(D)$, we set $(a_1,a_2)=(0,x^*_2)$. Define $H_2=\{(x_1,x_2)\in H:\;x_1\in(a_1, x^{*}_{1})\}$. One may easily show that
\begin{equation}\label{m5E}
\mbox{for any } (x_1,x_2)\in H_2 \mbox{ there exists a unique } y_2>0  \mbox{ such that \ref{m4B}) holds true.}
\end{equation}
The function $\hat y_2:(a_1,x^*_1)\to(0,\infty)$ given by (\ref{m5E}) is increasing and
\[
\lim_{x_1\to (x^{*}_{1})^-} \hat y_2(x_1)=+\infty
\]
\[\lim_{x_1\to a_1^+}\hat y_2(x_1)=\left\{\begin{array}{ll}
0& \mbox{ if $(C)$ holds, } \\
y_a &\mbox{ if $(D)$ holds, }
\end{array}\right.\]
where $y_a>0$ solves
\[
\bar{x}_2+A\bar{y}_2-Ay_a-CD=x^{*}_{2}.
\]
We now claim that there exists $E_{y-ext}(R)>0$ such that 
\begin{equation}\label{a1-less-than-b1}
a_1<b_1 \ \ \ \mbox{ if } \ \ \  0<E< E_{y-ext}(R). 
\end{equation}
To prove it first apply Lemma \ref{L1} to see that the intersection point of the parabolas $P_1(\bar y_1)$, that is the one given by \eqref{R1}, and $P_2(\bar y_2)$ (the one given by \eqref{R2}) intersect in a point that is not below the hyperbola $H$, i.e. it is not inside the bounded connected components of the set $(0,+\infty)\times (0, +\infty)\setminus H$. By Lemma \ref{L2} and the fact that $A\bar y_2>0$ we see that the parabola \eqref{R1} and $P_2(0)$ intersect in a point that is above $H$, i.e. inside the unbounded connected component of $(0,+\infty)\times (0, +\infty)\setminus H$.\\ 
\indent Now we notice that both $x^{*}_1$ and $x^{*}_2$ increase as $E$ does, on the other hand $A \bar y_1$ decreases and becomes negative after $E$ goes above $E_{y-ext}$. For $E=E_{y-ext}$, one has $A\bar y_1 =0$ and $P_1(0)=P_1(A\bar y_1)$, in particular $P_1(0)$ and $P_2(0)$ intersect above $H$ for $E=E_{y-ext}$. For any $E$, denote the positive intersection point of $P_1(0)$ and $P_2(0)$ by $(c_1(E), c_2(E))$ and observe that as $E\geq E_{y-ext}$ grows $(c_1(E),c_2(E))$ moves downward on the parabola $P_2(0)$ (independent of $E$) towards the hyperbola $H$, which, due to Lemma \ref{L3}, goes upward as $E$ increases. There a value that we denote by $E_{y-ext}(R)$ such that the point $(c_1 (E_{y-ext}(R)),c_2(E_{y-ext}(R)))$ reaches $H$ and, for $E<E_{y-ext}(R)$, $(c_1(E), c_2(E))$ stays above $H$, which proves \eqref{a1-less-than-b1}.

\indent We can also infer from the above considerations that, for any $E>E_{y-ext}(R)$, $y_1>0$ and $y_2>0$, the parabolas $P_1(y_1)$ and $P_2(y_2)$ intersect below $H$, which means that there are no positive solutions of \eqref{m4}. \\
\indent Now, for $E<E_{y-ext}(R)$, define $h:(a_1,b_1)\to \R$ by $h(x_1)=\hat y_2(x_1)/ \hat y_1(x_1)$. Clearly $h$ is contiunuous and, since $\hat y_1$ is decreasing and $\hat y_2$ is increasing, we infer that $h$ increasing.\\
\indent We also claim that
\begin{equation}\label{h-limits}
\lim_{x_1\to a_1^+} h(x_1)=0,  \ \ \  \lim_{x_1\to b_1^-} h(x_1)=\infty.
\end{equation}
To prove the claim consider the four cases below.\\
\indent If $(A)$ and $(C)$ hold, then $0\leq a_1$, $ b_1\leq x^{*}_{1}$ and
\[
\lim_{x_1\to a_1^+}  h(x_1)=
\left\{\begin{array}{lll}
\frac{0}{\infty}=0 & \mbox{ if }&a_1=0,\\
\frac{0}{\hat y_1(a_1)}=0 &\mbox{ if }& a_1>0,
\end{array}\right.  \ \ \
\lim_{x_1\to b_1^-} h(x_1) =
\left\{\begin{array}{lll}
\frac{\infty}{0^+}=\infty & \mbox{ if }&b_1=x^{*}_{1},\\
\frac{\hat y_2(b_1)}{0^+}=\infty &\mbox{ if }& b_1<x^{*}_{1}.
\end{array}\right.
\]
\indent If $(A)$ and $(D)$ hold, then $a_1=0$, $b_1\leq x^{*}_{1}$ and
\[
\lim_{x_1\to 0^+} h(x_1)=\frac{y_b}{\infty}=0, \ \ \ \lim_{x_1\to b_1^-} h(x_1)=
\left\{\begin{array}{lll}
\frac{\hat y_2(b_1)}{0^+}=+\infty & \mbox{ if }& b_1< x^{*}_1\\
\frac{+\infty}{0^+}=+\infty &\mbox{ if }& b_1=x^{*}_1.
\end{array}\right.
\]
\indent If $(B)$ and $(C)$ hold, then $0\leq a_1$, $b_1=x^{*}_{1}$ and
\[
\lim_{x_1\to (x^{*}_{1})^{-}} h(x_1)=\frac{+\infty}{y_a}=+\infty, \ \ \ \lim_{x_1\to a_1^+} h(x_1)=
\left\{\begin{array}{lll}
\frac{0}{y_b}=0 & \mbox{ if }&a_1>0,\\
\frac{0}{\infty}=0 &\mbox{ if }& a_1=0.
\end{array}\right.
\]
\indent If $(B)$ and $(D)$, then  $a_1=0$, $b_1=x^{*}_{1}$ and
\[
\lim_{x_1\to 0^+} h(x_1)=\frac{y_b}{+\infty}=0,\ \ \ \lim_{x_1\to (x^{*}_{1})^-} h(x_1)=\frac{+\infty}{y_a}=\infty.
\]
\indent In order to complete the proof we consider the equation
\begin{equation}\label{m4C2}
\frac{\hat y_2(x_1)}{\hat y_1(x_1)}=\frac{\bar{x}_1+B-x_1}{B}.
\end{equation}
The function on the left side is increasing and the function on the right side is decreasing on the interval $(a_1,b_1)$ and
\[
\lim_{x_1 \to a_1^+} \frac{y_2(x_1)}{y_1(x_1)}=0<\frac{\bar{x}_1+B-a_1}{B}  \ \ \ \mbox{ and }  \ \ \ \frac{\bar{x}_1+B-a_1}{B}<\lim_{x_1 \to b_1^-} \frac{y_2(x_1)}{y_1(x_1)}=\infty.
\]
Thus there exists a unique solution $\tilde{x}_1\in(a_1,b_1)$ to (\ref{m4C2}). If we take $\tilde{x}_2>0$ such that $(\tilde{x}_1,\tilde{x}_2)\in H$ and  set $\tilde{y}_i=\hat y_i(\tilde{x}_1)$ for $i=1,2$.
Then $(\tilde{x}_1,\tilde{x}_2,\tilde{y}_1,\tilde{y}_2)$ is the unique solution to the system of the equations \eqref{m4A}-\eqref{m4D}. 
\end{proof} 

\subsection{Stability of positive steady states}
If $(\tilde x_1, \tilde x_2, \tilde y_1, \tilde y_2)$ is the positive steady-state solution of \eqref{m2}, then 
\eqref{m2} can be rewritten (like in the passage between \eqref{r1A} and \eqref{r1B}) as 
\begin{equation}\label{m4a}
\left\{
        \begin{array}{l}
\dot x_1=-x_1 \big( a (x_1-\tilde{x}_1)+ b (y_1-\tilde{y}_1)\big)+m_x R(\tilde{x}_1x_2-\tilde{x}_2x_1)/\tilde{x}_1  \\
\dot x_2=-x_2\big(a  (x_2-\tilde{x}_2)+ b (y_2-\tilde{y}_2)\big)-m_x (1-R)(\tilde{x}_1 x_2-\tilde{x}_2 x_1)/\tilde{x}_2  \\
\dot y_1=c y_1   (x_1-\tilde{x}_1)+m_yR (\tilde{y}_1 y_2-\tilde{y}_2y_1)/\tilde{y}_1\\
\dot y_2 = c y_2   (x_2 -\tilde{x}_2)- m_y (1-R)(\tilde{y}_1y_2-\tilde{y}_2 y_1)/\tilde{y}_2.
\end{array} \right.
\end{equation} 
\begin{prop}\label{Jacobian}
The Jacobian matrix of the right hand side of \eqref{m4a} at the point $(\tilde{x}_1,\tilde{x}_2,\tilde{y}_1,\tilde{y}_2)$ is given by
\[
\left[
\begin{array}{cccc}
   -m_x R\tilde{x}_2/\tilde{x}_1- a \tilde{x}_1 & m_x R&- b \tilde{x}_1 & 0\\
   m_x (1-R)&-m_x (1-R)\tilde{x}_1/\tilde{x}_2- a \tilde{x}_2 & 0 & - b \tilde{x}_2\\
    c \tilde{y}_1 & 0 & -m_y  R \tilde{y}_2/\tilde{y}_1 & m_y  R\\
   0& c \tilde{y}_2 &  m_y  (1-R) & -m_y  (1-R) \tilde{y}_1/\tilde{y}_2
\end{array}
\right]
\]
and the coefficients of the characteristics polynomial $\lambda^4+A\lambda^3+B\lambda^2 + C\lambda+D$ are positive.
\end{prop}
\noindent We have computed the coefficients $A, B, C$ and $D$ by use of the MAPLE and Python symbolic computation libraries to see that all their symbolic components are positive.
\begin{remark}\label{stability-algorithm} {\em (i)} By Routh Theorem the steady state $(\tilde{x}_1,\tilde{x}_2,\tilde{y}_1,\tilde{y}_2)$ is locally stable if the coefficients $A,\,B,\,C,\,D$ are positive (which follows from Proposition \ref{Jacobian})  and 
\[
ABC>C^2+A^2D.
\]
The last inequality we verify numerically in all the simulations from Section 5.\\
\indent {\em (ii)} One may check that the idea from Proposition \ref{1S-MPA-stablility} for indicating a (global) Lyapunov functional applied to the steady state $(\tilde x_1, \tilde x_2, \tilde y_2, \tilde y_2)$ for the system \eqref{m4a} appears not to work in general.\end{remark}

\section{Effects of establishing an MPA in a prey-predator model}
In the section we present results concerning the yields and conservation effects of an MPA creation in the prey-predator model given by (\ref{m1A}). We shall consider a scenario with parameters
\begin{eqnarray}\label{parameters-values}
a=1.0, \ \ b=0.6, \ \ c=0.4, \ \ d=0.3,  \ \  q_x = 0.7, \ \ q_y = 0.35.
\end{eqnarray}
Before establishing an MPA the model is given by (\ref{pp1}). The natural equilibrium (with no fishing) is $\bar x=0.75$, $\bar y=0.42$ (see (\ref{pp2}) for $E=0$), i.e. the prey  biomass is $1.8$ times the predator biomass (they are of the same order of magnitude). By Proposition \ref{pp2}, the predator extinction  threshold $E_{y-ext}$ and the prey extinction threshold $E_{x-ext}$ are given by:
\begin{eqnarray*}
E_{y-ext}=0.15873, \ \ E_{x-ext}=1.42857.
\end{eqnarray*}
The total catch $C(E)$ given by (\ref{pp4}) attains its maximal value $0.25$ for the fishing effort 
$E=E_{x-ext}/2=0.714286$, which is far more than $E_{y-ext}$ and for this value of fishing effort the predator becomes extinct, which means that the catch consists of prey only. The total catch at the predator extinction threshold $E_{y-ext}$ is $0.098765$ (note that in this case it is around $2.5$ times less than the maximal value $C(E_{x-ext}/2)$). The predator catch $C_{Pred}(E)$ given by (\ref{pp3a}) has maximum at $E=E_{y-ext}/2$ and its value amounts to $C_{max}^{\, pred} = 0.005787$. 

\subsection{The effects for an over-exploited fishery}

We shall present the effects of establishing an MPA for the fishing effort equal to the threshold $E=E_{y-ext}$. In the framework of the prey-predator model we may consider it as an over-exploited fishery, since the predator population is subject to extinction for that value of fishing effort. Recall, that by Theorem \ref{EQ1.1}, the positive steady-state
$$
(\tilde{x}_1(R),\,\tilde{x}_2(R),\,\tilde{y}_1(R),\,\tilde{y}_2(R))
$$
of the system (\ref{m1A}) is unique for each $R\in(0,1)$. We shall measure the conservation effect of establishing an MPA by the use of the {\em normalized total biomass of predator} $\bar{y}_{nom}(R)$ given by
\[
\bar{y}_{nom}(R)=\frac{\tilde{y}_1(R)+\tilde{y}_2(R)}{\bar{y}},
\]
where $\bar y$ is the predator biomass at the nontrivial equilibirium without any fishing effort and in consequence without MPA ($R=0$, $E=0$). This quantity shows the total predator population size in relation to the natural environmental equilibrium.\\
\indent The yield effect of an MPA creation for a given fishing effort $E$ will be measured by means of the {\em normalized total catch} $C_{nom}(R)$ and the {\em normalized predator catch} $C_{nom}^{pred}(R)$ defined by
\[
C_{nom}(R) =\frac{Eq_x\tilde{x}_1(R)+Eq_y\tilde{y}_1(R)}{C(E_{y-ext})},\;\;\; \; \;
C_{nom}^{\, pred} (R)=\frac{Eq_y\tilde{y}_1(R)}{C_{max}^{\, pred}},
\]
where 
$$
C(E_{y-ext}) = E_{y-ext}\, q_y \bar x(E_{y-ext}) 
$$
with $\bar x(E_{y-ext})$ given by \eqref{pp1a} (recall $\bar y (E_{y-ext}) = 0$), and 
$C_{max}^{\, pred}=\max_{E>0} C^{pred}(E)$, i.e. the maximal predator catch that is attained without any MPA ($R=0$).  \\
\indent In the particular case of the parameters given by \eqref{parameters-values}, the biomass of predator at the natural equilibrium is $\bar y=0.42$, the catch with the fishing effort equal to the predator extinction threshold is $C(E_{y-ext}) = 0.098765$ and the predator maximal catch is $C_{max}^{\, pred} = C(E_{y-ext}/2) = 0.005787$. In order to illustrate the effects of implementing an MPA when the fishing efforts approaches the predator extinction threshold, we assume that $E=E_{y-ext}$ and compute the normalized values $\bar{y}_{nom}(R)$, $C_{nom}(R)$ and $C_{nom}^{\, pred}(R)$ for different sizes $R\in (0,1)$. We considered nine scenarios of mobility coefficients: $m_x,m_y\in\{0.1,\,1,\,10\}$, where, like in the single species model, the value $0.1$ of the mobility coefficient corresponds to sedentary species and values $1$ and $10$ correspond to slowly migrating and moderately fast migrating species, respectively. To reduce the number of considered scenarios we skipped in the section the cases with $m_x=100$ or $m_y =100$ that correspond to fast migrating species. Like in the one species model (see Figure \ref{Fig1}) for moderately fast migrating and for fast migrating species the results are very close. In Figures \ref{F6A}, \ref{F6B} and \ref{F6C} we present the graphs of the normalized quantities $\bar{y}_{nom} (\cdot)$, $C_{nom}(\cdot)$ and $C_{nom}^{\, pred}(\cdot)$ in for $m_x=1$ and $m_y$ taking values $0.1$, $1$ and $10$. The case with $m_x$ equal to $0.1$ or $10$ are alomost identical to the one with $m_x=1$.\\ 
\indent Here the equilibrium points $(\tilde x_1 (R), \tilde x_2 (R), \tilde y_1(R), \tilde y_2 (R))$ were determined numerically by finding $\tilde x_1(R)$ as a solution of \eqref{x1-key-equation},
$\tilde x_2(R)$ is obtained from \eqref{x-1-x-2-relation} and $\tilde y_1(R), \tilde y_2(R)$
from the formulas \eqref{formulas-for_y1-y2}. The local stability of these equilibrium points is verified as indicated in Remark \ref{stability-algorithm}.


\begin{figure}[ht]
  \centering  \includegraphics[keepaspectratio=true,width=7cm]{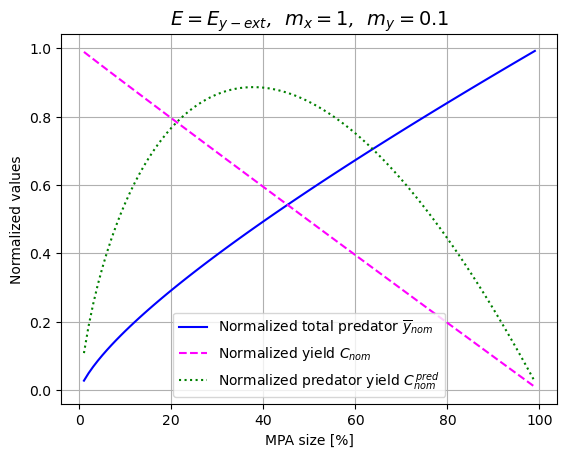}
  \caption{Normalized total predator population, total yield, total predator yield and in dependence of $R$ for $E=E_{y-ext}$ $m_x=1$, $m_y=1$},\label{F6A}
\end{figure}
\begin{figure}[ht]
  \centering  \includegraphics[keepaspectratio=true,width=7cm]{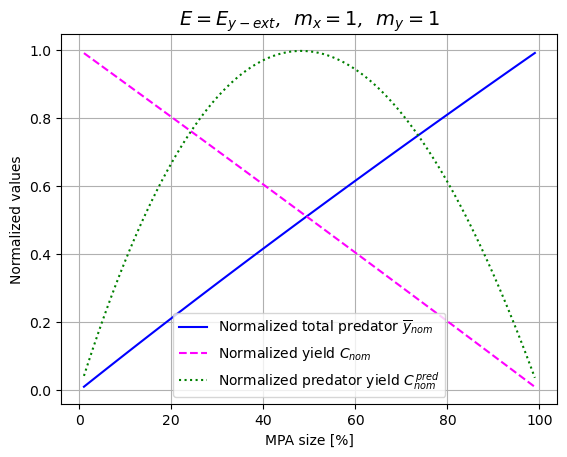}
  \caption{Normalized total predator population, total yield, total predator yield and in dependence of $R$ for $E=E_{y-ext}$ $m_x=1$, $m_y=1$},\label{F6B}
\end{figure}
\begin{figure}[ht]
  \centering  \includegraphics[keepaspectratio=true,width=7cm]{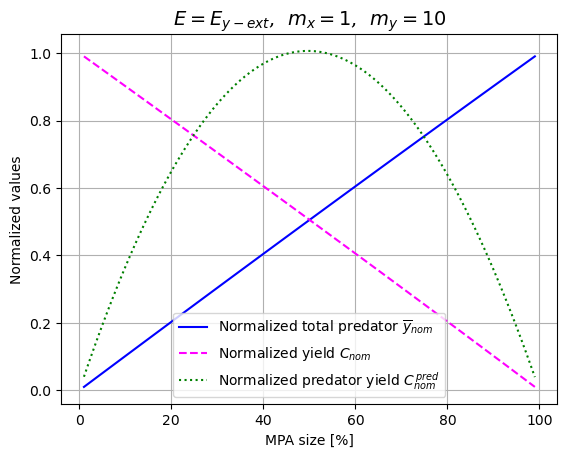}
  \caption{Normalized total predator population, total yield, total predator yield and in dependence of $R$ for $E=E_{y-ext}$ $m_x=1$, $m_y=1$},\label{F6C}
\end{figure}
\indent The conservation benefits of establishing an MPA are observed for every considered values of the mobility coefficients. The normalized total predator population $\bar{y}(\cdot)$ appears to be an increasing function (nearly/almost linear with the slope $1$). The yield effects are diverse. The total yield functions $C_{nom}$ are decreasing (nearly/almost linear with the slope $-1$). Whereas the predator catch $C_{nom}^{\, pred}$ has a maximum for the MPA sizes between $35\%-50\%$, depending on $m_x$ and $m_y$. For sedentary predator ($m_y=0.1$) the maximal predator catch is attained near $R=40\%$ and it realizes almost $90\%$ of the maximal predator catch $C_{max}^{\, pred}$. For slowly and moderately fast migrating predator, i.e. for $m_y=1$ and $m_y=10$, respectively, the predator catch has its maximum near $R=50\%$ and it realizes almost $100\%$ of the maximal predator catch.

\subsection{Conservation benefits for $E>E_{y-ext}$} 
Let us consider the case when for any reason the fishing effort $E$ exceeds the threshold $E_{y-ext}$. To this end we take $E=1.25\cdot E_{y-ext}$ and look at the dependence of $\bar{y}_{nom}$, $C_{nom}$ and $C_{nom}^{\, pred}$ on $R$ for $m_x=1$ and $m_y\in \{ 0.1, 1, 10\}$ (the graphs for $m_x$ taking the values $0.1$ and $10$ are very similar). The results are shown in Figures \ref{F7A}, \ref{F7B} and \ref{F7C}.\\
\begin{figure}[ht]
  \centering  \includegraphics[keepaspectratio=true,width=7cm]{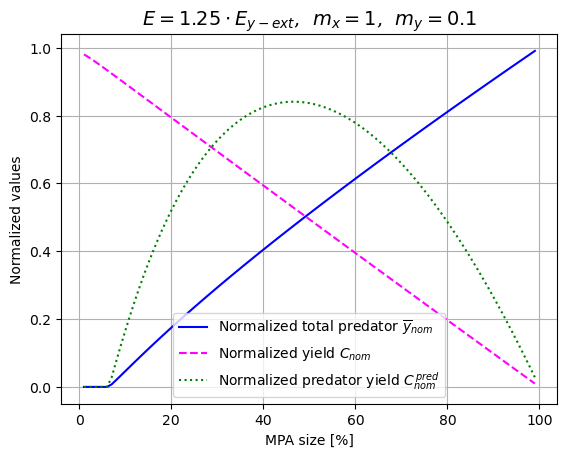}
  \caption{Normalized total predator population, total yield, total predator yield and in dependence of $R$ for $E=1.25 \cdot E_{y-ext}$ $m_x=1$, $m_y=0.1$},\label{F7A}
\end{figure}
\begin{figure}[ht]
  \centering  \includegraphics[keepaspectratio=true,width=7cm]{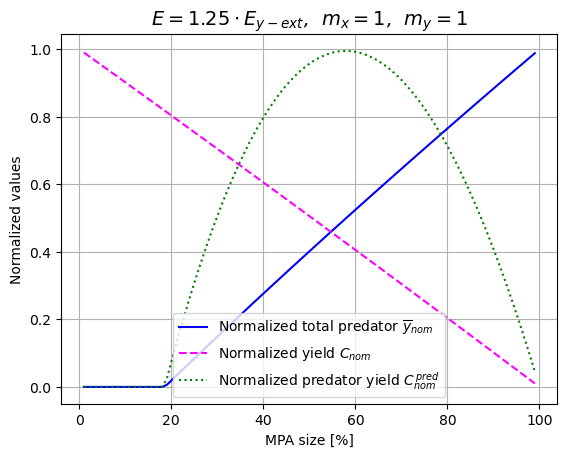}
  \caption{Normalized total predator population, total yield, total predator yield and in dependence of $R$ for $E=1.25\cdot E_{y-ext}$ $m_x=1$, $m_y=1$},\label{F7B}
\end{figure}
\begin{figure}[ht]
  \centering  \includegraphics[keepaspectratio=true,width=7cm]{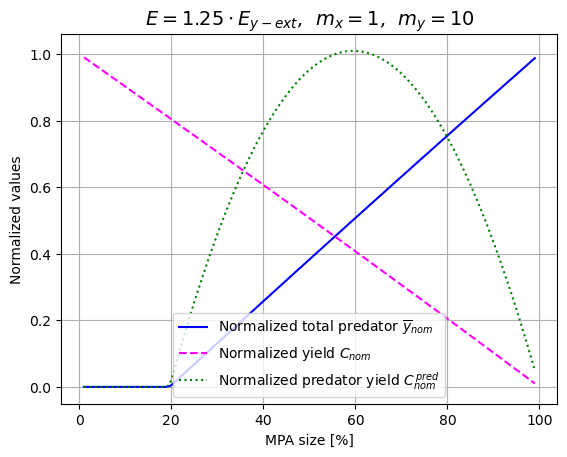}
  \caption{Normalized total predator population, total yield, total predator yield and in dependence of $R$ for $E=1.25 \cdot E_{y-ext}$ $m_x=1$, $m_y=1$},\label{F7C}
\end{figure}
\indent Observe that in all the considered situations when the MPA is not sufficiently large and the fishing effort remains at the level exceeding $E_{y-ext}$ (here $E=1.25\cdot E_{y-ext}$), the predator must extinct. However, if the MPA size $R$ is around $30\ \%$, then the total predator population exceeds or is close to $20 \%$ of that natural one $\bar y$ (with no fishing effort). Although the total catch decreases with the gradient $-1$, the total predator catch exceeds $80\ \%$
of the maximal one $C_{max}^{\ pred}$ (with no MPA) for $R$ near $50\%$ if $m_y=0.1$ and reaches $100\%$ of $C_{max}^{\ pred}$ for $R$ near $50\%$ if $m_y=1$ or $m_y=10$. Hence, the implementation of an MPA allows for the fishing effort even considerably exceeding the predator extinction threshold, in the sense that it maintains the predator at a survival level and assure decent predator yield quantity while the total yield is always lower by the percent equal to the MPA size.

\subsection{Shift of predator extinction threshold for fishing effort}
The previous subsection shows that implementing an MPA raises the fishing effort value for which the predator extincts completely. In this subsection we shall see how much the extinction threshold changes. Let for a given $R\in (0,1)$, $E_{y-ext}(R)$ be as in Theorem \ref{EQ1.1}, that is the lowest value of $E$, for which the predator dies out, i.e. $\tilde y_1(R)=\tilde y_2 (R)=0$. Analogously, we may consider $E_{x-ext}(R)$. We saw in Proposition \ref{pp2} that, in case without MPA, $E_{y-ext}>E_{x-ext}$, that is it is the predator that extincts first when $E$ grows. It will not change when an MPA is implemented, since due to Remark \ref{extinction-order}, $E_{y-ext}(R)\leq E_{x-ext}(R)$ for $R\in (0,1)$.\\
\indent Let us consider the system \eqref{m1A} with the parameters given by \eqref{parameters-values} and various $E$ in a few scenarios depending on $m_x$ and $m_y$. In order to illustrate to which extent creating an MPA protects the predator from extinction, we determine the value, which we call the {\em normalized threshold predator extinction effort} 
$$
E_{y-ext} (R)/E_{y-ext}
$$ 
for various $R\in (0,1)$. In Figures \ref{F11A}, \ref{F11B} and \ref{F11C} we show the graphs of the normalized threshold predator extinction effort for different $m_x$ and $m_y$ (we have skipped the graph for $m_x=0.1$ and $m_y=0.1$ since then the values are significantly larger than those for $m_x=0.1$ and $m_y=1$). Observe that, in all the considered cases, when the MPA size is around $R=30\%$ the predator extinction threshold $E_{y-ext}(R)$ is more than $1.4\cdot E_{y-ext}$, which indicates that an MPA indeed protects the predator population even if $E$ goes beyond the extinction threshold $E_{y-ext}$. 

\begin{figure}[ht]
  \centering  \includegraphics[keepaspectratio=true,width=7cm]{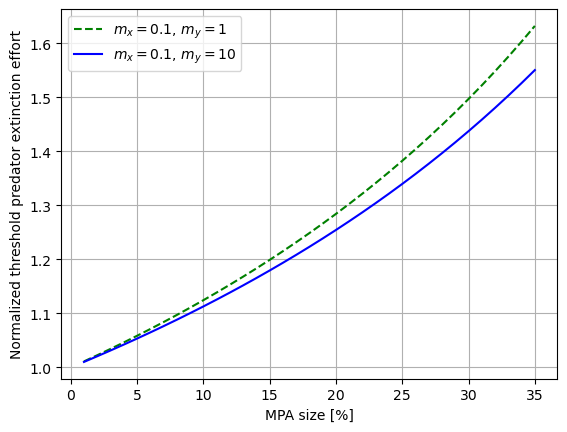}
  \caption{Dependence of the normalized threshold value (i.e. the lowest value at which the predator extincts divided by $E_{y-ext}$) on MPA size} \label{F11A}
\end{figure}
\begin{figure}[ht]
  \centering  \includegraphics[keepaspectratio=true,width=7cm]{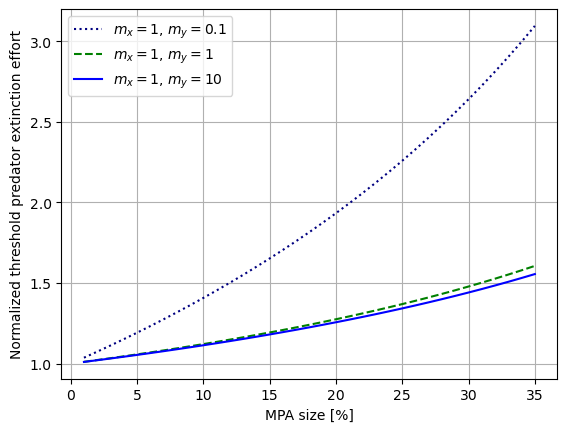}
  \caption{Dependence of the normalized threshold value (i.e. the lowest value at which the predator extincts divided by $E_{y-ext}$) on MPA size}
  \label{F11B}
\end{figure}

\begin{figure}[ht]
  \centering  \includegraphics[keepaspectratio=true,width=7cm]{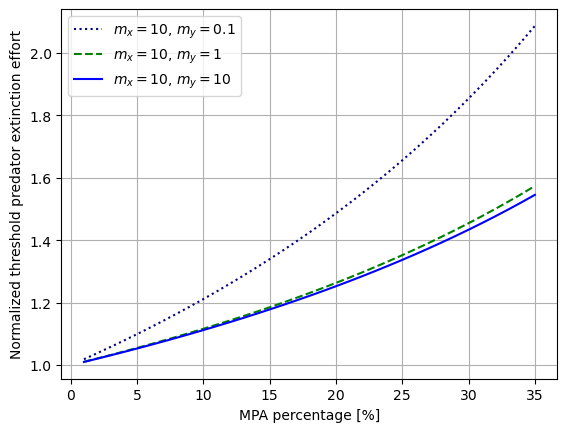}
  \caption{Dependence of the normalized threshold value (i.e. the lowest value at which the predator extincts divided by $E_{y-ext}$) on MPA size}
  \label{F11C}
\end{figure}
\section{Concluding remarks}

Our analysis of the single-species and prey-predator models indicates that a need for implementing an MPA appears in two different contexts but both are referred to as over-exploitation. A fishery with a single species is considered over-exploited if the fishing effort ($E$) surpasses the species extinction threshold. However, in the prey-predator model, over-exploitation is identified when the fishing effort reaches the predator extinction threshold, which notably is lower than that of the prey. When fishing effort surpasses the predator extinction threshold, the fishery simply becomes a single-species one. Thus, in models featuring two species, the primary objective of establishing an MPA is the conservation of both species.\\
\indent We observe that in over-exploited fisheries, whether governed by the single-species or prey-predator model, the establishment of an MPA of reasonable size (e.g. $30\%$) yields significant conservation benefits, ensuring that the populations are sustained at secure levels. Furthermore, when the MPA is sufficiently large (around $50\%$), the total yield in the single species model and the total predator yield in the prey-predator model are sustainable and reaching $80\%$ or even $100\%$ of their maximum (depending on species mobility). However, in the prey-predator model, preserving the predator comes with a trade-off. In the two-species model, the total catch, accounting for both prey and predator, decreases as the size of the MPA increases. This mirrors the trend observed in the single-species model, where an MPA consistently results in decreased yields unless fishing effort exceeds the optimal level. Therefore, a similar pattern in the prey-predator model is unsurprising as the fishing effort which poses a threat to extinction in the single-species model is significantly higher than the effort that endangers predator extinction in the prey-predator model. Then the value of the fishing effort corresponding to the over-exploitation in the prey-predator model is at the level of sustainable fishing effort for the prey. But 
the prey are the main part of the catch and for a sustainable fishing effort a marine reserve reduces the catch, hence the total catch drop is caused by the decrease in the prey catch that is not compensated by the increase of the predator catch.

\end{document}